\documentclass[12pt]{article}
\usepackage{amsmath,amsthm,latexsym,amssymb,amsfonts,amsbsy,mathrsfs,bm,braket,color}
\usepackage{graphicx,multirow,tabularx}
\usepackage{amsfonts,delarray}
\usepackage[comma]{natbib}
\usepackage{url} 

\usepackage{fancyhdr}
\usepackage{epstopdf}
\newcommand{\blind}{0}

\numberwithin{equation}{section}

\DeclareMathOperator*{\var}{var}

\DeclareMathOperator*{\Ex}{E}

\newcommand{\bz}{\bm z}
\newcommand{\bZ}{\bm Z}
\newcommand{\bd}{\bm d}
\newcommand{\bh}{\bm h}
\newcommand{\mM}{\mathcal M}
\newcommand{\mX}{\mathcal X}
\newcommand{\mY}{\mathcal Y}

\newcommand{\bbR}{\mathbb R}

\newcommand{\bbT}{\mathrm{T}}
\newcommand{\bbZ}{\mathrm{Z}}
\newcommand{\blambda}{\bm \lambda}

\newtheorem{thm}{Theorem}[section]
\newtheorem{cor}{Corollary}[section]
\newtheorem{lem}{Lemma}[section]
\theoremstyle{remark}
\newtheorem{rmk}{Remark}[section]

\begin{document}

\def\spacingset#1{\renewcommand{\baselinestretch}%
{#1}\small\normalsize} \spacingset{1}


\if0\blind
{
  \title{\bf Eigen-Adjusted Functional Principal Component Analysis}
  \author{Ci-Ren Jiang\\
    Institute of Statistical Science, Academia Sinica, \\
    \\
    Eardi Lila \\
    Department of Biostatistics, University of Washington, Seattle, \\
    \\
    John AD Aston\thanks{
    \textit{j.aston@statslab.cam.ac.uk}}\hspace{.2cm} \\
    Statistical Laboratory, University of Cambridge, \\
    \\
    Jane-Ling Wang \\
    Department of Statistics, University of California, Davis}
  \maketitle
} \fi

\if1\blind
{
  \bigskip
  \bigskip
  \bigskip
  \begin{center}
    {\LARGE\bf Title}
\end{center}
  \medskip
} \fi

\bigskip
\begin{abstract}
Functional Principal Component Analysis (FPCA) has become a widely-used dimension reduction tool for functional data analysis. When additional covariates are available, existing FPCA models integrate them either in the mean function or in both the mean function and the covariance function. However, methods of the first kind are not suitable for data that display second-order variation, while those of the second kind are time-consuming and make it difficult to perform subsequent statistical analyses on the dimension-reduced representations. To tackle these issues, we introduce an eigen-adjusted FPCA model that integrates covariates in the covariance function only through its eigenvalues. In particular, different structures on the covariate-specific eigenvalues -- corresponding to different practical problems -- are discussed to illustrate the model's flexibility as well as utility. To handle functional observations under different sampling schemes, we employ local linear smoothers to estimate the mean function and the pooled covariance function, and a weighted least square approach to estimate the covariate-specific eigenvalues. The convergence rates of the proposed estimators are further investigated under the different sampling schemes. In addition to simulation studies, the proposed model is applied to functional Magnetic Resonance Imaging scans, collected within the Human Connectome Project, for functional connectivity investigation.
\end{abstract}

\noindent%
\vfill

\newpage
\spacingset{1.5} 
\section{INTRODUCTION}
\label{sec:intro}

Principal component analysis is a classical dimension reduction tool in multivariate statistical analysis and its extension to functional data, termed Functional Principal Component Analysis (FPCA), plays a central role in the analysis of samples that are curves, functions, or surfaces as shown in the survey article by \cite{Shan:14}. However, most of the existing FPCA approaches \citep[see, e.g.,][]{Rao:58, DauxPR:82, RiceS:91, Card:00, JameHS:00, RiceW:01, YaoMW:05:1, HallMW:06,LiH:10:1, ChenJ:17} assume that the observed samples arise from the same population and do not accommodate information from a set of additional covariates. In this paper, motivated by the analysis of spatio-temporal brain imaging data, we introduce a novel FPCA model where additional covariates, modeling for instance spatial locations, affect both the mean and covariance of the functional samples describing the temporal component of the data.

So far, relatively little of the literature has covered FPCA methodology that adapts to the covariate information. Some of the proposed methods merely integrate the covariate information in the systematic part, i.e. the mean function \cite[see, e.g.,][]{ChioMW:03,JianW:11,ZhanPW:13,ZhanW:14}. In order to ease the computational burden and alleviate the curse of dimensionality, in \cite{ChioMW:03, ZhanPW:13} and \cite{ZhanW:14} additive structures for the mean function are considered, while in \cite{JianW:11} a single index model for the mean is adopted. Here, we refer to these approaches as mean-adjusted FPCA. Alternatively, both the systematic and stochastic parts, i.e. the mean function and the covariance structure, are assumed to vary with the covariates \cite[see, e.g.,][]{Card:06,JianW:10,LiSB:14}. Specifically, \cite{Card:06,JianW:10} assume that the covariance function varies with the covariates via both its eigenfunctions and eigenvalues. \cite{JianAW:09} extend the covariate-adjusted FPCA model in \cite{JianW:10} to a multiplicative model for Positron Emission Tomography image analysis. Moreover, \cite{LiSB:14} propose a model that accounts for the covariate effect in skewed functional data. We refer to these methods as fully-adjusted FPCA.

In some applications, such as functional brain imaging, the covariance structure of functional data describing temporal brain activity is a key element to study brain connectivity. However, mean-adjusted FPCA models are not designed to integrate covariate information in the covariance function, although they are attractive because of their computational efficiency. On the other hand, fully-adjusted FPCA approaches are time-consuming if the dimension of the covariates is not low and they make subsequent statistical analyses on the dimension-reduced output difficult due to the fact that they return a set of covariates-specific eigenvalues and eigenfunctions, which are not easily comparable across covariate values.

Therefore, in this work, we will extend the PCA model in \cite{Flur:84, Flur:86} to functional data and term it eigen-adjusted FPCA. Specifically, the eigen-adjusted FPCA model assumes the covariance function varies with the covariates via its eigenvalues while the corresponding eigenfunctions remain independent of the covariates. The advantages of the proposed model are two-fold. In practical applications, FPCA is generally the first step of a subsequent statistical analysis. For instance, in the proposed application to brain imaging data in Section~\ref{sec:application}, the ultimate goal is to find a parcellation that clusters brain locations (modeled as covariates) with distinct functional connectivity patterns (i.e. distinct covariance structures). For this purpose, the covariate-independent eigenfunctions of the eigen-adjusted FPCA model offer a common reference frame that enables comparison and statistical analysis of the associated covariate-specific eigenvalues. Moreover, eigen-adjusted FPCA while being more flexible than mean-adjusted FPCA maintains comparable computational times. This is particularly important for big data applications, such as the aforementioned brain imaging studies. By means of simulations, we demonstrate that eigen-adjusted FPCA can also be useful under a fully-adjusted generative model, where applying fully-adjusted FPCA might not be computationally feasible. Other related work can be found in \cite{BoenPR:02,BenkH:05,BenkHK:09,BoenRS:10,CoffHDH:11}.


Both functional data and longitudinal data can be modeled as observations from stochastic processes, but they are different in their sampling schemes. Specifically, functional data are densely and regularly recorded while longitudinal data are sparsely and irregularly observed. FPCA has been classically concerned with the analysis of densely observed functional samples. More recently, the methodology has been extended to deal with sparsely observed longitudinal data, where only a few repeated measurements are available for each sample, invalidating approaches based on nonparametric reconstructions of individual functions \citep[see, e.g.,][]{YaoMW:05:1, HallMW:06}. In this work, we develop a unified eigen-adjusted FPCA framework for both types of data.

The rest of this paper is organized as follows. In Section~\ref{sec:model} we introduce the eigen-adjusted FPCA model. The proposed estimators and their asymptotical properties are provided in Section~\ref{sec:estimation}. Section~\ref{sec:simulations} consists of simulation studies that demonstrate the finite sample performance of the proposed method. The eigen-adjusted FPCA model is applied to analyze brain functional connectivity in Section~\ref{sec:application}. Conclusions and discussions are provided in Section~\ref{sec:discussions}. Moreover, the assumptions are given in Appendix~\ref{app:assum}, and an alternative estimation approach for the covariate-specific eigenvalues is provided in Appendix~\ref{app:PC}. The proofs are provided in the supplementary materials.

\section{COVARIATE-ADJUSTED FPCA}\label{sec:model}

Traditionally, FPCA handles random trajectories that are functions of time (or space). Ignoring any covariate information for the moment, let $X(t)$ be a stochastic process in $L^2$ with mean function $\mu(t)$ and covariance function $\Gamma(s,t)$ where $s,t\in\bbT$ and $\bbT$ is a finite compact interval in $\bbR$. FPCA is equivalent to a spectral decomposition of $\Gamma$ and leads to the well-known Karhunen-Lo\`{e}ve decomposition of $X$,
\begin{equation*}
X(t) = \mu(t) + \sum_{j=1}^\infty A_j\varphi_j(t),
\end{equation*}
where $\varphi_j$ is the eigenfunction of $\Gamma$ associated to the $j$-th largest eigenvalue, and $A_j = \langle X-\mu,\varphi_j\rangle$ is the $j$-th functional principal component (PC) score. Here $\langle \cdot,\cdot\rangle$ stands for the inner product in the $L^2$ space, i.e., $\langle a,b\rangle = \int_\bbT a(t)b(t)dt$ for $a,b\in L^2(\bbT)$.

When covariate information is available, $X$ is viewed as a random function of both time $t$ and covariate $\bz$. Specifically, $X(t,\bz)$ is a random function with mean $\mu(t,\bz)$ and covariance function $\Upsilon(s,t,\bz)$, where $s,t\in\bbT$ and $\bz$ is in a $p$-dimensional compact hypercube $\bbZ \subset \bbR^p$. To accommodate $\bz$ into the framework of FPCA, \cite{JianW:10} proposed two models: fully-adjusted FPCA (fFPCA) and mean-adjusted FPCA (mFPCA). The difference between them is how the covariance structure is handled. In fFPCA, it is assumed that there exists an orthogonal expansion of $\Upsilon(s,t,\bz)$ in terms of eigenfunctions $\psi_j(t,\bz)$ and non-increasing eigenvalues $\eta_j(\bz)$, i.e., $\Upsilon(s,t,\bz) = \sum_{j=1}^\infty \eta_j(\bz)\psi_j(s,\bz)\psi_j(t,\bz)$, and thus one can represent $X$ as
\begin{equation*}
X(t,\bz) = \mu(t,\bz) + \sum_{j=1}^\infty A_j(\bz) \psi_j(t,\bz),
\end{equation*}
where $A_j = \langle X-\mu, \psi_j \rangle$ is the $j$-th functional PC score with mean zero and variance $\eta_j(\bz)$. In mFPCA, $\bz$ is treated as a realization of the random variable $\bZ$. Ignoring $\bZ$ after centering leads to a pooled covariance
\begin{equation*}
\Gamma^*(s,t) = \Ex \{\Upsilon(s,t,\bZ)\} = \int_\bbZ \Upsilon(s,t,\bz)dG_{\bz}(\bz),
\end{equation*}
where $G_{\bz}(\cdot)$ is the distribution function of $\bZ$. Assume $\Gamma^*$ is a smooth function and there exists an orthogonal expansion in terms of eigenfunctions $\psi_j^*(t)$ and nonincreasing eigenvalues $\eta_j^*$; i.e., $\Gamma^*(s,t) = \sum_{j=1}^\infty \eta_j^*\psi_j^*(s)\psi_j^*(t)$, and the random function $X$ can be represented as
\begin{equation*}
X(t,\bz) = \mu(t,\bz) + \sum_{j=1}^\infty A_j^* \psi_j^*(t),
\end{equation*}
where $A_j^* = \langle X-\mu, \psi_j^* \rangle$ is the $j$-th functional PC score with mean zero and variance $\eta_j^*$. In mFPCA, $\Gamma^*$ can be estimated with a lower-dimensional smoother and thus the estimator is of a faster convergence rate. On the contrary, fFPCA is a more flexible model as the covariate information is used in estimating the covariance function as well as the eigenvalues and the eigenfunctions.

However, FPCA rarely represents the final stage of a statistical analysis as, no matter how complex, an FPCA model will unlikely represent a faithful description of the underlying generative model. FPCA is instead commonly used as a dimension reduction tool, where the subsequent analysis is performed on the PC scores. So, despite the fFPCA model being more flexible, it introduces complications when using it as a dimension reduction tool, as the scores $A_j(\bz)$ do not give a complete description of the stochastic process at the covariate value $\bz$; in fact, the scores are expressed with respect to a reference frame $\psi_j(t,\bz)$ which is also dependent on $\bz$, invalidating comparison across covariate values.

Therefore, in the following subsection, we propose a model that is flexible enough to deal with dependencies of the covariance functions on covariate information, yet representing an effective and useful dimension reduction tool.

\subsection{Eigen-Adjusted FPCA}
Assume that the covariance function $\Gamma(s,t,\bz)$ has an orthogonal expansion in terms of eigenfunctions $\phi_k(t)$ and nonnegative eigenvalues $\lambda_k(\bz)$. Specifically,
\begin{equation}\label{eq:cov}
\Gamma(s,t,\bz) = \sum_{k=1}^\infty \lambda_k(\bz)\phi_k(s)\phi_k(t),
\end{equation}
where $\lambda_1(\bz)> \lambda_2(\bz) > \cdots \geq 0$ 
and $\sum_{k=1}^\infty \lambda_k(\bz)<\infty$ for $\bz\in\bbZ$ and $s,t\in\bbT$. By the Karhunen-Lo\'{e}ve expansion, $X(t,\bz)$ can be represented as
\begin{equation}\label{eq:cfpca}
X(t,\bz) = \mu(t,\bz) + \sum_{k=1}^\infty A_{k}(\bz) \phi_k(t),
\end{equation}
where $A_k(\bz)=\langle X-\mu,\phi_k\rangle$ is the $j$-th functional PC score with mean zero and variance $\lambda_k(\bz)$. Compared to mFPCA, Model (\ref{eq:cfpca}) is more general in that $\Gamma$ in (\ref{eq:cov}) can vary with $\bz$, and compared to fFPCA, it is computationally more efficient in estimating the covariance function. The second argument will be further demonstrated in Section \ref{sec:estimation}.

The model of $\Gamma$ in (\ref{eq:cov}) allows for different structural assumptions on $\lambda_k(\bz)$, tailored to specific problems. Roughly speaking, there are three scenarios. The first scenario is that $\lambda_k(\bz)$ are grouped and may not vary smoothly with $\bz$. This setting was considered in most related FPCA literature [e.g., \cite{BenkHK:09} and \cite{BoenRS:10}], where the goal is to test the equality of covariance functions of different known groups. The second scenario is that $\lambda_k(\bz)$ are continuous and vary smoothly with $\bz$, such as when $\bz$ models spatial coordinates and we assume spatially smooth variation. The third scenario is that $\lambda_k(\bz)$ are piecewise smooth functions of $\bz$, i.e., $\lambda_k(\bz)$ are smooth functions within each group. The latter is a more realistic model for the analysis of the brain imaging data in our final application. 
In this paper, we focus on estimating $\lambda_k(\bz)$ under the second scenario as the third scenario can be dealt with by simply exploiting the group information. 
When the group structure is not known, our numerical studies suggest that estimating the eigenvalues as in the second scenario and applying a clustering approach to the estimated $\lambda_k(\bz)$ can help retrieve the group information.

\section{ESTIMATION}\label{sec:estimation}

Standard procedures to perform FPCA include (i) estimating the mean function, (ii) estimating the covariance function, (iii) estimating the eigenfunctions and eigenvalues, and (iv) predicting the functional PC scores. Under some regularity conditions on the mean and covariance functions, local linear smoothers \citep{FanG:96} can be applied to estimate them. Eigenfunctions and eigenvalues can be estimated via the application of an eigen-decomposition to the estimated covariance. Standard numerical approaches and PACE \citep{YaoMW:05:1} can be applied to predict the functional PC scores when data are dense and sparse, respectively. In step (ii), directly estimating $\Gamma(s,t,\bz)$ is computational demanding especially when both $\bbT$ and $\bbZ$ are multidimensional. Exploiting the special structure of $\Gamma(s,t,\bz)$ in (\ref{eq:cov}), to ease computational burden, we circumvent step (ii). Specifically, we apply an eigen-decomposition to the pooled covariance $\Gamma^*=\Ex (\Gamma)$ to obtain the eigenfunctions, and propose a Weighted Least Squares (WLS) approach, for both dense and sparse data, to estimate the eigenvalues that vary with $\bz$. Below we provide the details of the proposed estimators as well as their asymptotic properties.


\subsection{Mean Function}
Let $\bz_i\in\bbZ$ be the covariate of the $i$th subject, whose $j$th observation made at time $t_{ij}\in\bbT$ is
\begin{equation}\label{eq:cfpca2}
Y_{ij} = \mu(t_{ij},\bz_i) + \sum_{k=1}^\infty A_{ik} \phi_k(t_{ij}) + \varepsilon_{ij},
\end{equation}
where $\varepsilon_{ij}$ is the independent measurement error with mean zero and variance $\sigma^2$, for $j=1,\ldots,N_i$ and $i=1,\ldots,n$. Theoretically, any $(p+1)$-dimensional smoother can be applied to estimate $\mu$. Here, we use a $(p+1)$-dimensional local linear smoother and denote the estimator as $\hat\mu$. 
Specifically,
\begin{align*}
 \hat\mu(t,\bz) & = \hat b_0, \text{ where }\\ &  (\hat b_0,\hat b_1,\hat{\bm b}_2)^T  = \arg\min_{\bm{b}}  \frac{1}{n}\sum_{i=1}^n \frac{1}{N_i} \sum_{j=1}^{N_i}  \left\{ Y_{ij}-b_0-b_1(t_{ij}-t)-{\bm b}_2^T(\bz_i-\bz) \right\}^2 \\ & \hspace*{6cm} \times K_{h_t}(t_{ij}-t)\Big( \prod_{k=1}^p K_{h_z^{(k)}}(z_i^{(k)}-z^{(k)})\Big),
\end{align*}
$\bm{b} = (b_0,b_1,\bm b_2^T)^T$, $K_h(\cdot) = K(\cdot/h)/h$, $K$ is a kernel function defined in Assumption A.2 in Appendix A, and $h_t$ and $\bh_z = (h_z^{(1)},\ldots,h_z^{(p)})^T$ are the bandwidths for $\bbT$ and $\bbZ$, respectively. For simplicity, 
we assume that $h_z^{(k)}$'s are all of the same order as $h_z$. Let $\gamma_{nk} = \left( n^{-1}\sum_{i=1}^n N_i^{-k} \right)^{-1}$ for $k=1,2$, and $\delta_n = \left[\left\{1+1/(\gamma_{n1} h_t)\right\}\log n/(nh_z^p)\right]^{1/2}$; denote $h_1 \approx h_2$ if $h_1$ is of the same order as $h_2$, and $h_1 \lesssim h_2$ (resp. $h_1\gtrsim h_2$)  if $h_1$ is of smaller (resp. larger) order than $h_2$. Below we provide the asymptotical properties of $\hat\mu$.

\begin{thm} \label{A:thm1}
Assume that A.1-A.2 and B.1-B.2 hold. Then,
\begin{equation}\label{eq:as-muhat}
\sup_{t\in\bbT,\bz\in\bbZ} |\hat{\mu}(t,\bz)-\mu(t,\bz)| = O(h_t^2+h_z^2 + \delta_n) \ a.s..
\end{equation}
\end{thm}

In (\ref{eq:as-muhat}), $h_t^2+h_z^2$ and $\delta_n^2$ are the order of bias and that of variance, respectively, due to smoothing. We further elaborate on the convergence rates of $\hat\mu$ under two different sampling schemes in the following corollary.
\begin{cor} \label{A:cor1}
Assume that A.1-A.2 and B.1-B.2 hold. \\
(a) 
If $\max_{1\leq i\leq n} N_i \leq \mM$ for some $\mM<\infty$ and $h_t\approx h_z \approx h$,
\begin{equation*}
\sup_{t\in\bbT,\bz\in\bbZ} |\hat{\mu}(t,\bz)-\mu(t,\bz)| = O\big(h^2 + \{\log n/(nh^{p+1})\}^{1/2}\big) \ a.s..
\end{equation*}
The optimal convergence rate is $\left(\log n/n\right)^{2/(p+5)}$.

(b) If $\max_{1\leq i\leq n} N_i \geq \mM_n$, where $\mM_n^{-1} \approx h_t \lesssim h_z \approx (\log n/n)^{1/(p+4)} $ is bounded away from zero,
\begin{equation}
\sup_{t\in\bbT,\bz\in\bbZ} |\hat{\mu}(t,\bz)-\mu(t,\bz)| = O\big( (\log n/n)^{2/(p+4)}\big) \ a.s..
\end{equation}
\end{cor}

\begin{rmk}
(a) represents the case of longitudinal data where $N_i$ is finite and so it is reasonable to have $h_t$ and $h_z$ of the same order.  (b) represents functional data where the observations are intensely recorded. Since $N_i\rightarrow\infty$, $h_t$ is of smaller order than $h_z$.
\end{rmk}

\subsection{Pooled Covariance and Its Eigenfunctions}
In practice, the eigenfunctions are estimated by performing an eigen-decomposition on the discretized covariance function. However, directly estimating $\Gamma$ may be too computationally demanding especially when both $\bbT$ and $\bbZ$ are multi-dimensional. Thus, we propose to estimate the eigenfunctions $\phi_k$ by applying an eigen-decomposition to the pooled covariance function,
\begin{equation}\label{eq:pcov}
\Gamma^*(s,t) = \Ex\{\Gamma(s,t,\bZ)\} = \sum_{k=1}^\infty \Ex\{\lambda_k(\bZ)\} \phi_k(s)\phi_k(t),
\end{equation}
where $\Gamma(s,t,\bZ)$, $\lambda_k(\bZ)$ and $\phi_k(t)$ are defined in (\ref{eq:cov}). Any two-dimensional smoother can be applied to estimate $\Gamma^*$ and we pick a local linear smoother here. Specifically,
 \begin{align*}
& \hat\Gamma^*(s,t)  = \hat b_0 \text{, where } \\
& (\hat b_0,\hat b_1,\hat b_2)^T  = \arg\min_{\bm{b}} \frac{1}{n}\sum_{i=1}^n \frac{1}{N_i(N_i-1)} \sum_{j\neq k}^{N_i} \left\{ \hat U_{ij} \hat U_{ik}-b_0-b_1(t_{ij}-s)-b_2(t_{ik}-t) \right\}^2 \\
& \hspace{6cm} \times K_{h_\Gamma}(t_{ij}-s)K_{h_\Gamma}(t_{ik}-t),
\end{align*}
$\bm{b} = (b_0,b_1,b_2)^T$, $\hat U_{ij}=Y_{ij}-\hat\mu(t_{ij},\bz_i)$, $K$ is a kernel function defined as that for $\hat\mu$, and $h_\Gamma$ is the smoothing bandwidth. The diagonal terms $\{\hat U_{ij}\hat U_{ij}\mid 1\leq j \leq N_i, 1\leq i \leq n\}$ are removed while estimating $\Gamma^*$ since $\text{cov}(Y_{ij},Y_{ik}) = \Gamma(t_{ij},t_{ik},\bz_i)+ \sigma^2 \delta_{(t_{ij}=t_{ik})},$ where $\delta_{(t_{ij}=t_{ik})}=1$ if $t_{ij}=t_{ik}$ and 0 otherwise. For simplicity, we let 
\begin{equation*}
\delta_{n1}  = \left\{\left(1+\frac{1}{\gamma_{n1} h_\Gamma}\right)\frac{\log n}{n}\right\}^{1/2}\text{, and }
\delta_{n2}  = \left\{\left(1+\frac{1}{\gamma_{n1} h_\Gamma}+ \frac{1}{\gamma_{n2}h_\Gamma^2}\right)\frac{\log n}{n}\right\}^{1/2}.
\end{equation*}
\begin{thm} \label{A:thm2}
Assume that A.1-A.2 and B.1-B.4 hold. We can obtain
\begin{align*}
\sup_{s,t\in\bbT} |\hat{\Gamma}^*(s,t)-\Gamma^*(s,t)|  & = O\big(\delta_{n2}+h_\Gamma^2+\delta_{n1}(h_t^2+h_z^2+\delta_n)+h_t^4+h_z^4+\delta_n^2\big)  \ a.s., \\
& = O(\delta_{n2}+h_\Gamma^2+h_t^4+h_z^4+\delta_n^2)  \ a.s..
\end{align*}
\end{thm}

\begin{rmk} Note that $\delta_{n1}(h_t^2+h_z^2+\delta_n)$ should be of smaller order than $\delta_{n2}$  as $h_t$, $h_z$ and $\delta_n$ go to zero as $n\rightarrow\infty$.
When $p\leq 3$ (in most real examples and in our data), it is reasonable to believe that $\delta_n^2 \lesssim \delta_{n2}$ and we elaborate the convergence rates in the following corollary under such a condition. Even so, some brief discussions on the convergence rates when $p\geq 4$ will be given in Section 3.4.
\end{rmk}

\begin{cor} \label{A:cor2}
Assume that A.1-A.2 and B.1-B.4 hold. \begin{itemize}
\item[(a)] If $\max_{1\leq i\leq n} N_i \leq \mM$ for some $\mM<\infty$ and $\{(\log n/n) h_\Gamma^2\}^{1/(2p+2)} \lesssim h_t, h_z \lesssim h_\Gamma^{1/2}$,
\begin{equation}
\sup_{s,t\in\bbT} |\hat{\Gamma}^*(s,t)-\Gamma^*(s,t)| = O\left( h_\Gamma^2 + \{\log n/(nh_\Gamma^2)\}^{1/2}\right) \ a.s..
\end{equation}
The restriction of bandwidths holds automatically when $p=1$, and it leads to $h_\Gamma \gtrsim ( \log n/n )^{1/(p-1)}$ when $p\geq 2$. The optimal convergence rate is $( \log n/n )^{1/3}$.


\item[(b)] If $\max_{1\leq i\leq n} N_i \geq \mM_n$, where $\mM_n^{-1} \lesssim h_\Gamma \approx h_t \lesssim (\log n/n)^{1/4}$ is bounded away from zero,
\begin{equation}
\sup_{s,t\in\bbT} |\hat{\Gamma}^*(s,t)-\Gamma^*(s,t)| = O\left(h_z^4+ \delta_n^2+(\log n/n)^{1/2}\right) \ a.s..
\end{equation}
The optimal convergence rate is $( \log n/n )^{1/2}$ given $h_z \lesssim (\log n/n)^{1/8}$. 
\end{itemize}
\end{cor}

Once $\Gamma^*$ is estimated, $\hat\phi_k$ can be obtained through
$$
\int_\bbT \hat\Gamma^*(s,t) \hat\phi_k(s) ds = \hat\lambda^*_k \hat\phi_k(t).
$$
Below we provide the asymptotic properties of $\hat\phi_k$ and those of $\hat \lambda^*_k$ for $1\leq k\leq L$.

\begin{thm}\label{A:thm3}
Assume that A.1-A.2 and B.1-B.4 hold, for $1\leq j\leq L$:
\begin{align*}
& \|\hat\phi_j -\phi_j\| = O(h_\Gamma^2+\delta_{n1}+h_t^4+h_z^4+\delta_n^2) \ a.s.. \\
& |\hat\lambda_j^* - \lambda_j^*|  = O\big((\log n/n)^{1/2} +h_\Gamma^2+\delta_{n1}(h_t^2+h_z^2+\delta_n)+h_t^4+h_z^4+\delta_n^2\big) \ a.s..
\end{align*}
\end{thm}

Again, we elaborate on the convergence rates of $\hat\phi_j$ in regards to different sampling schemes.

\begin{cor} \label{A:cor3}
Assume that A.1-A.4 and B.1-B.3 hold. \begin{itemize}
\item[(a)] If $\max_{1\leq i\leq n} N_i \leq \mM$ for some fixed $\mM$ and $\{(\log n/n)h_\Gamma\}^{1/(2p+2)} \lesssim h_t, h_z \lesssim h_\Gamma^{1/2} $,
\begin{equation}
\|\hat\phi_j -\phi_j\| = O\left( h_\Gamma^2 + \{\log n/(nh_\Gamma)\}^{1/2} \right) \ a.s..
\end{equation}
The optimal convergence rate is $(\log n/n)^{2/5}$.


\item[(b)] If $\max_{1\leq i\leq n} N_i \geq \mM_n$, where $\mM_n^{-1} \lesssim h_\Gamma \approx h_t \lesssim (\log n/n)^{1/4}$ is bounded away from zero,
\begin{equation}
\|\hat\phi_j -\phi_j\| = O\left((\log n/n)^{1/2}+h_z^4+\delta_n^2\right) \ a.s..
\end{equation}
The optimal convergence rate is $( \log n/n )^{1/2}$ given $h_z \lesssim (\log n/n)^{1/8}$. 
\end{itemize}
\end{cor}

\subsection{Covariate-specific Eigenvalues}

Given that $\Ex(A_{ik}^2) = \lambda_k(\bz_i)$, $\lambda_k(\bz)$ can be estimated via applying a $p$-dimensional smoother to $\{(A_{ik}^2,\bz_i)\mid 1\leq i\leq n\}$. However, this is only feasible for dense data. When the data are sparse, such as longitudinal observations, the PC scores can not be predicted accurately via numerical approaches for integration, and instead PACE \citep{YaoMW:05:1} is usually considered. However, applying a $p$-dimensional smoother to the squared PC scores predicted by PACE may not be appropriate as $\var\{\Ex(A_{ik}|\bm y_i)\} \leq \var(A_{ik})=\lambda_k(\bz_i)$, where $\Ex(A_{ik}|\bm y_i)$ is the PACE predictor of $A_{ik}$.
To solve this issue, we propose a WLS procedure to estimate the eigenvalues $\lambda_k(\bz)$ for both sparse and dense data. An alternative, exclusively for dense data, is provided in Appendix \ref{app:PC} and is named as the PC-based approach.

Let $\mX_i = \{ \phi_\ell(t_{i,j})\phi_\ell(t_{i,k})\}$, where $j<k$, be a $[N_i(N_i-1)/2]\times L_n$ matrix, $\mY_i = (C_{i,12},\ldots, C_{i,(N_i-1)(N_i)})^T$, $C_{i,jk}=U_{ij}U_{ik}$, $U_{ij} = Y_{ij}-\mu(t_{ij},\bz_i)$, and $\blambda_z = (\lambda_1(\bz),\ldots,\lambda_{L_n}(\bz))^T$. Thus, $\mY_i$ can be represented as
\begin{equation*}
\mY_i = \mX_i \blambda_{z_i} + \bm\epsilon_i,
\end{equation*}
where $\bm\epsilon_i$ is the remainder term. Given that $\lambda_k$'s are smooth functions of $\bz$, it is thus reasonable to estimate $\blambda_z$ through
\begin{equation} \label{eq:lda_wls}
\hat{\blambda}_z^W = (\hat{\mX}^T\mathcal{W}_z\hat{\mX})^{-1}(\hat{\mX}^T\mathcal{W}_z\hat{\mY}),
\end{equation}
where $\hat{\mY} = (\hat{\mY}_1^T,\ldots,\hat{\mY}_n^T)^T$, $\hat{\mX} = (\hat{\mX}_1^T,\ldots,\hat{\mX}_n^T)^T$,
$\hat{\mX}_i = \{ \hat\phi_\ell(t_{i,j})\hat\phi_\ell(t_{i,k})\}_{[N_i(N_i-1)/2]\times L_n}$,  $\hat{\mY}_i = (\hat C_{i,12},\ldots, \hat C_{i,(N_i-1),N_i})^T$, $\hat C_{i,jk}  =  \hat U_{i,j} \hat U_{i,k}$, $\mathcal{W}_z = \text{diag}(\bm w_1^T,\ldots,\bm w_n^T)$, and $\bm w_i$ is a vector with $N_i(N_i-1)/2$ identical elements, $\prod_{k=1}^p K_{h_\lambda^{(k)}}(z_i^{(k)}-z^{(k)})$.

\begin{rmk}
Given that the order of $\hat\lambda_k(\bz)$ is determined by that of $\hat\lambda^*_k$, the estimated eigenvalues of the pooled covariance, the assumption $\lambda_1(\bz)>\lambda_2(\bz)>\ldots \geq 0$ for $\bz\in\bbZ$ in Model (\ref{eq:cov}) implies that the correct order is expected. However, what really matters is that these estimated covariate-specific eigenvalues correspond to the same order of estimated eigenfunctions (a common reference frame) allowing the comparison and statistical analysis of these covariate-specific eigenvalues. The assumption $\lambda_1(\bz)>\lambda_2(\bz)>\ldots \geq 0$ for $\bz\in\bbZ$ could be slightly relaxed; specifically, the order of $\lambda_k(\bz)$ could vary with $\bz$ as long as there are no tie for $\lambda_k^*$ to ensure the estimate of $\phi_k$.
\end{rmk}

Here $L_n$ could be a slowly divergent sequence if the following condition on $\lambda_k(\bz)$ is satisfied:
\begin{equation}\label{B.7}
\lambda_j(\bz)>\lambda_{j+1}(\bz)>0, E(A_{j}^4) \leq c \lambda_j(\bz) \text{, and } \lambda_j(\bz) -\lambda_{j+1}(\bz) > c^{-1}j^{-(a_1+1)},
\end{equation}
for $\bz\in \bbZ$, $a_1>1$, and some $c>0$. A similar assumption can be found in the literature of functional regression (e.g., \cite{HallH:07} and \cite{YaoLW:15}). 

\begin{thm} \label{A:thm4}
Under assumptions A.1--A.2, and B.1--B.7, we have
\begin{align} \label{eq:ldaz}
\sup_{\bz\in\bbZ} | \hat\blambda_z^W - \blambda_z | & = O\big(L_n^{-a_1}+L_n h_\lambda^2+L_n^{2}(h_\Gamma^2+h_t^4+h_z^4+\delta_{n2}+\delta_n^2) \nonumber \\
& \hspace{15mm} +h_t^2+h_z^2+\delta_{n}+1/(nh_\lambda^p)^{1/2}\big) \ a.s..
\end{align}
\end{thm}

Generally, $(h_t+h_\Gamma)\lesssim (h_z+h_\lambda)$ and we can obtain the following corollary for sparse data.
\begin{cor} \label{A:cor4}
Under assumptions A.1--A.2, and B.1--B.7,
\begin{equation} \label{eq:ldaz2}
\sup_{\bz\in\bbZ} |\hat\blambda_z^W - \blambda_z| = O\big(L_n^{-a_1} + L_n h_\lambda^2+ L_n^{2}(h_\Gamma^2 +\delta_{n2}) \big) \ a.s..
\end{equation} 
If further $h_\Gamma^2 \approx \delta_{n2}$, i.e., $ h_\Gamma \approx (\log n/n)^{1/6}$, $L_n h_\lambda^2 \approx L_n^{2}h_\Gamma^2$, and $L_n^{2}h_\Gamma^2\approx L_n^{-a_1}$, the optimal convergence rate is $(\log n/n)^{a_1/3(a_1+2)}$.
\end{cor}

\subsection{When $p\geq 4$}
Corollaries \ref{A:cor2} -- \ref{A:cor4} provide the optimal convergence rates for $p\leq 3$. The optimal convergence rates for a general $p$ might be of interest to other examples and below we provide some discussions.
\begin{itemize}
\item[--] \textbf{Extension of Corollary \ref{A:cor2}}\\ The optimal convergence rate of (a) can only be achieved when $p\leq 7$. When $p\geq 8$ and $h_\Gamma^{1/2} \lesssim h_t, h_z \lesssim \{(\log n/n)h_\Gamma^2\}^{1/(2p+2)}$,  the optimal convergence rate is $(\log n/n )^{4/(p+5)}$. Under the assumptions of case (b), the optimal convergence rate is $(\log n/n)^{4/(4+p)}$ and $h_z \approx (\log n/n)^{1/(4+p)}$ when $p\geq 4$.
\item[--] \textbf{Extension of Corollary \ref{A:cor3}}\\ Under the assumptions of case (a), the optimal convergence rate, $( \log n/n )^{2/5}$, is achieved when $p\leq 5$. When $p>5$, the optimal convergence rate is $(\log n/n)^{4/(p+5)}$ if the bandwidths satisfy $h_\Gamma^{1/2} \lesssim h_t, h_z \lesssim \{(\log n/n)h_\Gamma \}^{1/(2p+2)}$. Under the assumptions of case (b) and and $h_z \approx ( \log n/n)^{1/(4+p)}$, the optimal convergence rate is $( \log n/n)^{4/(4+p)}$  when $p\geq 4$.
\item[--] \textbf{Extension of Corollary \ref{A:cor4}}\\ The optimal convergence rate is achieved when $p\leq 7$. When $p\geq 8$,
the order of (\ref{eq:ldaz2}) becomes $O(L_n h_\lambda^2+ L_n^{2}(h_z^4 +\delta_{n}^2)+ L_n^{-a_1}+\delta_{n})$. If further $h_z^4 \approx \delta_{n}^2$ (i.e., $h_z \gtrsim (\log n/n)^{1/(p+5)}$), $L_n h_\lambda^2 \approx L_n^2 h_z^4$ and $L_n h_z^4 \approx L_n^{-a_1}$ and $L_n^2\delta_n^2 \gtrsim \delta_n$, $L_n \gtrsim (n/\log n)^{3/2(p+5)}$ and the convergence rate is $(\log n/n)^{3a_1/2(p+5)}$.
\end{itemize}

\section{SIMULATION STUDIES}\label{sec:simulations}

This section consists of three simulation studies. Simulation 1 compares the two eigenvalue estimators (the PC-based approach and the proposed WLS approach) under the second scenario with two sampling schemes. Simulation 2 shows that when estimating the covariance function $\Upsilon(s,t,\bz) = \sum_{k=1}^\infty \tau_k(\bz)\psi_k(s,\bz)\psi_k(t,\bz)$ is not computationally feasible, approximating $\Upsilon$ with the common covariance function $\Gamma(s,t,\bz) = \sum_{k=1}^\infty \lambda_k(\bz)\phi_k(s)\phi_k(t)$ is a better choice than doing it with the pooled covariance function $\Gamma^*(s,t) = \sum_{k=1}^\infty \lambda^*_k\phi_k(s)\phi_k(t)$. Simulation 3 comprises four settings of data generation under the third scenario and is aimed to show that the proposed WLS estimator helps reveal the latent group information of each function and thus significantly improves clustering accuracy. In the following numerical studies, we employed PACE version 2.17 (available at \url{http://www.stat.ucdavis.edu/PACE/}) to estimate the eigenfunctions.


\subsection{Simulation 1}

For each run we generated $n$ (=200 and 400) curves from model (\ref{eq:cfpca2}). For $i=1,\ldots,n$, we let $z_i\sim U(0,1)$, $\mu(t,z_i)=0$, $\phi_1(t)=-\cos(\pi t/10)/\sqrt{5}$, $\phi_2(t)=\sin(\pi t/10)/\sqrt{5}$, $A_{ik}(z_i) \sim N(0,\lambda_k(z_i))$ for $k=1,2$, where $\lambda_1(z_i)=4\{1+2\sin(0.1+\pi z_i^2/2)\}$ and $\lambda_2(z_i)= 2\{2+\sin(2z_i\pi)\}$, and $\varepsilon_{ij}\sim N(0,1)$. For complete data, the observations were made at 51 equally spaced time points over $[0,10]$ per curve. To generate sparse data, we randomly selected $N_i$ points out of the 51 equally spaced time points for complete data, where $N_i$ were randomly selected from $\{4,5,\ldots,10\}$ for each curve. Each simulation consists of 100 runs.

For fair comparisons, we let the bandwidth in (\ref{eq:lda_wls}) be $0.2$ for both types of data. 
In the PC-based approach, we first predicted $A_{ik}$ via the trapezoidal rule and PACE, for functional data and longitudinal data, respectively; next, we applied a $p$-dimensional smoother to $\{(\hat A_{i,k}^2, \bz_i), i=1,\ldots,n\}$ to estimate $\lambda_k(\bz)$. The integrated squared error (ISE) of the PC-based approach and that of the proposed WLS approach were summarized in Table \ref{mise_sim3}, indicating that the proposed WLS approach outperforms the PC-based approach under both sampling schemes. Interestingly, the proposed WLS approach outperforms the PC-based approach even in the complete data setting. Intuitively, this could be explained by the fact that, in the proposed WLS approach, ``more observations'' are employed to directly estimate the eigenvalues, i.e., $\sum_{i=1}^n N_i(N_i-1)/2$ (with $N_i=51$) v.s. $n$. This effect will disappear as $n$ gets much larger.

\begin{table}
\caption{\label{mise_sim3}Average ISE of the estimated eigenvalues with standard errors in the parentheses when $n=200$ and $n=400$.}
\centering\scriptsize
\begin{tabular}{ccccc|cccc} \hline
& \multicolumn{4}{c|}{Complete} & \multicolumn{4}{c}{Sparse} \\
  & \multicolumn{2}{c}{n=200} & \multicolumn{2}{c|}{n=400} & \multicolumn{2}{c}{n=200} & \multicolumn{2}{c}{n=400} \\ \cline{2-9}
Method  &  $\lambda_1(\bz)$  & $\lambda_2(\bz) $ & $\lambda_1(\bz)$  & $\lambda_2(\bz) $ &  $\lambda_1(\bz)$  & $\lambda_2(\bz) $ & $\lambda_1(\bz)$  & $\lambda_2(\bz) $  \\ \hline
 PC-based &  4.28 (3.45) & 0.99 (0.72)  & 1.99 (1.45) & 0.57 (0.38) & 5.63 (4.02) & 2.01 (1.30) & 3.29 (2.29) & 1.57 (0.82) \\
 WLS& 2.93 (2.39) & 0.75 (0.51) & 1.49 (1.07) & 0.47 (0.29) & 4.46 (3.66) & 1.79 (1.38)  & 2.49 (1.83) & 0.91 (0.60) \\ \hline
\end{tabular}
\end{table}

The supplementary material provides  an additional variation of simulation 1, which demonstrates the dependency of the estimation (ISE) on $p$ as shown in Section \ref{sec:estimation}.

\subsection{Simulation 2}

Understanding the covariance structure is essential in some studies, such as brain functional connectivity analyses, but directly estimating the covariance $\Upsilon(s,t,\bz) = \sum_{k=1}^\infty \eta_k(\bz)\psi_k(s,\bz)\psi_k(t,\bz)$ may not be computationally feasible. Here, we demonstrate that approximating $\Upsilon$ with $\Gamma(s,t,\bz) = \sum_{k=1}^\infty \lambda_k(\bz)\phi_k(s)\phi_k(t)$ leads to a more satisfactory output than doing it with $\Gamma^*(s,t) = \Ex\{\Upsilon(s,t,\bz)\}$. To mimic real data, we borrowed the {\it phantom} function in MATLAB to design a 2D spatial structure. Specifically, three groups were considered and shown in Figure \ref{Sim_phantom}, and the data were generated from the fFPCA model,
$$
X(t,\bz) = \mu(t,\bz)+\sum_{k=1}^2 A_k(\bz) \psi_k(t,\bz),
$$
where $\mu(t,\bz)=0$, $\psi_1(t,\bz) = \sqrt{2}\sin(2\pi \|\bz\|t)/2$, $\psi_2(t,\bz) = \sqrt{2}\cos(2\pi \|\bz\|t)/2$, $A_k(\bz) \sim N(0,\lambda_k(\bz))$, $\lambda_k(\bz)$ was listed in Table \ref{lda_k}, $\bz\in[0,1]\times[0,1]$ and $t\in[0,1]$. The observations were made at $128\times 128$ spatial locations and each location contained 31 observations made at equi-spaced time points. The data were further contaminated with measurement errors generated from $N(0,.2^2)$. The experiment consists of 100 runs.
\begin{figure}
       \centering
       \includegraphics[width=3in]{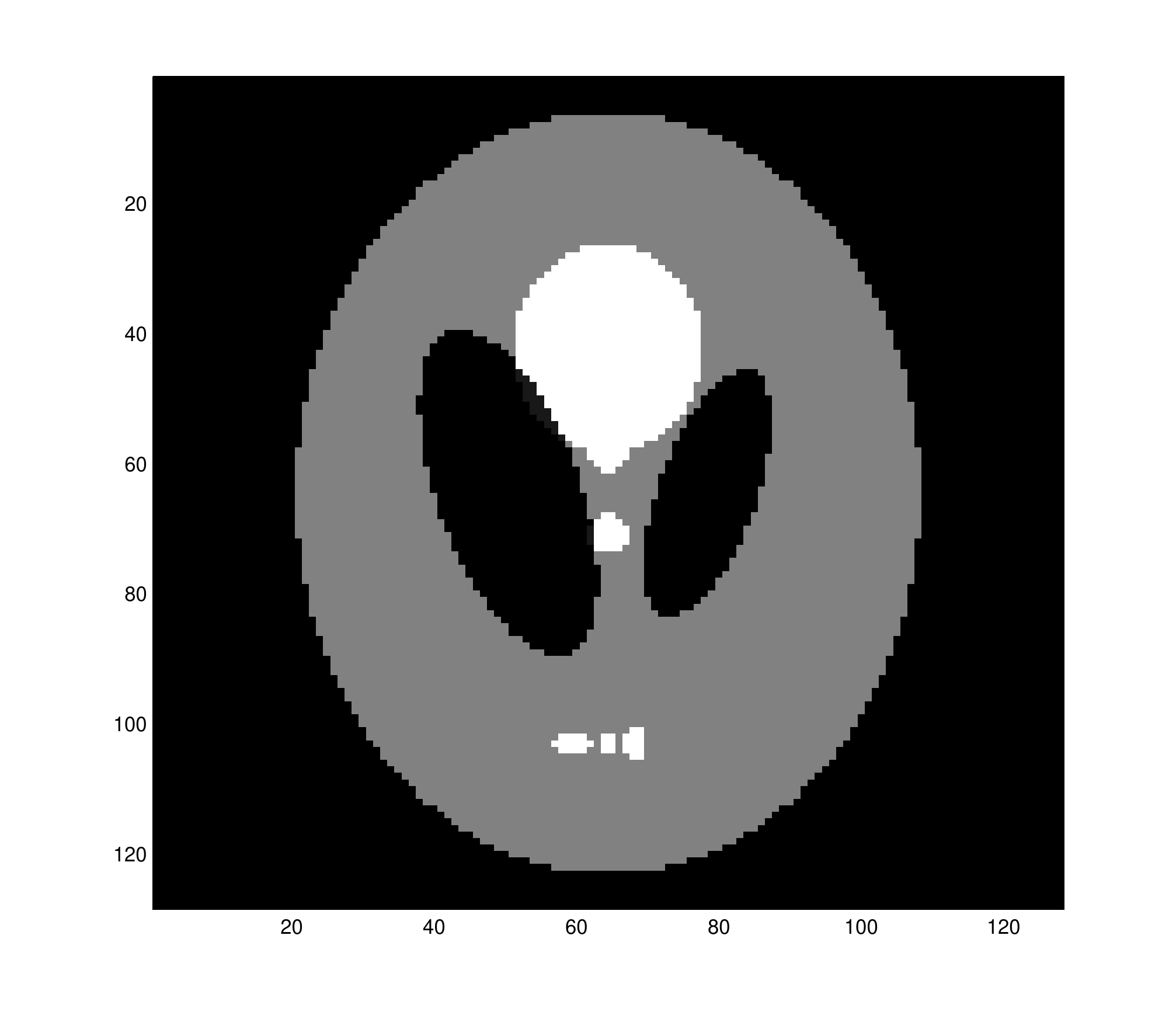}
       \caption{\label{Sim_phantom} Areas of three groups: the grey area corresponds to $\mathcal{S}_1$, the white area corresponds to $\mathcal{S}_2$, and the rest (black) area $\mathcal{S}_0$ is white noise.}
\end{figure}

\begin{table}
\caption{\label{lda_k} The $k$-th eigenvalue of group $i$, $\lambda_{i,k}(\bz)$. }
\centering
\begin{tabular}{cccc} \hline
cluster &  1  & 2 \\ \hline
 $\mathcal{S}_2 $  & $8+\cos(z^{(1)} 2 \pi)/2$ & $4+\sin((0.5+z^{(1)})2 \pi)/8$ \\
 $\mathcal{S}_1 $  & $3+ \cos(z^{(1)}\pi)\sin(0.5+z^{(2)})$ & $1.5+\cos(z^{(1)}\pi)\sin(0.5+z^{(2)})/2$ \\
 $\mathcal{S}_0 $  & 0 & 0  \\ \hline
\end{tabular}
\end{table}

To evaluate the performance of approximating $\Upsilon$ with $\Gamma$ and with $\Gamma^*$, we compared their ISEs, i.e., $\int\{\Upsilon(s,t,\bz)- \hat{\Gamma}(s,t,\bz)\}^2dsdt d\bz$ and $\int \{\Upsilon(s,t,\bz)- \hat{\Gamma}^*(s,t)\}^2dsdtd\bz$. The average ISEs (and the standard error) of $\hat{\Gamma}^*$ and $\hat{\Gamma}$ are $.3600$ $(.0002)$ and $.1369$ $(.0027)$, respectively. This shows that if the data are generated from a fFPCA model, approximating the covariance function with $\Gamma$ is overall a better choice given that it is computationally feasible and has a smaller ISE.

\subsection{Simulation 3}

To mimic real data, we employed the {\it phantom} structure again. We further assumed that the eigenvalues were group-dependent (listed in Table \ref{lda_k} and additional simulation). Two cases of the eigenfunctions were considered: in one, different groups shared a common set of eigenfunctions while in the other, the eigenfunctions were group-dependent. The goal was to demonstrate that clustering the WLS eigenvalue estimates outperforms clustering $\text{PC}^2$-based eigenvalue estimates, which is closely related to a mean-adjusted FPCA model, as it is not easy to verify this argument with real data in functional connectivity study in general. The $k$-means approach was employed for clustering and the performance was evaluated through the precision rates and the recall rates (defined in Table \ref{Sim3out}).

\subsubsection*{3A}
The data were generated from (\ref{eq:cfpca2}), where $\mu(\cdot,\cdot)=0$, $\lambda_{i,k}(\bz)$ were given in Table \ref{lda_k}, $\phi_k(t) = \sin(2\pi k t) + \cos(2\pi k t)$ for $k=1$ and $2$, and $\varepsilon_{ij}$ are \emph{i.i.d.} $N(0,0.4)$. 

\subsubsection*{3B}
The data generation setting was almost identical to that of 3A except that the eigenvalues were further smoothed across space. Specifically, we smoothed $\{(\lambda_k(\bz_i),\bz_i)|i=1,\ldots, 128^2\}$ in 3A via a Gaussian product kernel and the standard deviation for Gaussian kernel was $.03$. The purpose was to introduce smoothness across space. Due to the smoothing procedure, the correct clustering rate of a single location was modified accordingly. 

\subsubsection*{3C}
The procedure of generating data is similar to that of 3A with the difference that the eigenfunctions here are group-dependent. Specifically, the data were generated from (\ref{eq:cfpca2}), where $\mu(\cdot,\cdot)=0$, $\lambda_{i,k}(\bz)$ were given in Table \ref{lda_k}, $\varepsilon_{ij}$ were \emph{i.i.d.} from $N(0,0.4)$, and the group-dependent eigenfunctions were 
\begin{equation}\label{eq:sim1c_eig}
\phi_k(t,\bz) =  \left\{\begin{array}{lr}
        \phi_{1,k}(t) = \sin(2\pi kt) + \cos(2\pi kt), & \text{for } \bz\in \mathcal{S}_1,\\
        \phi_{2,k}(t) = \sin(2\pi kt) + \cos(4\pi kt), & \text{for } \bz\in \mathcal{S}_2,\\
        \phi_{3,k}(t) =  0, & \text{otherwise.}
        \end{array}\right.
\end{equation}


\subsubsection*{3D}
The procedure of data generation was almost identical to that of 3C except that the eigenvalues were further smoothed across space. Specifically, we smoothed $\{(\lambda_k(\bz_i),\bz_i)|i=1,\ldots, 128^2\}$ in 3C via a Gaussian product kernel and the standard deviation for Gaussian kernel was $.03$. The correct clustering rate of a single location was also modified accordingly. 

\begin{figure}
       \centering
       \includegraphics[width=6in]{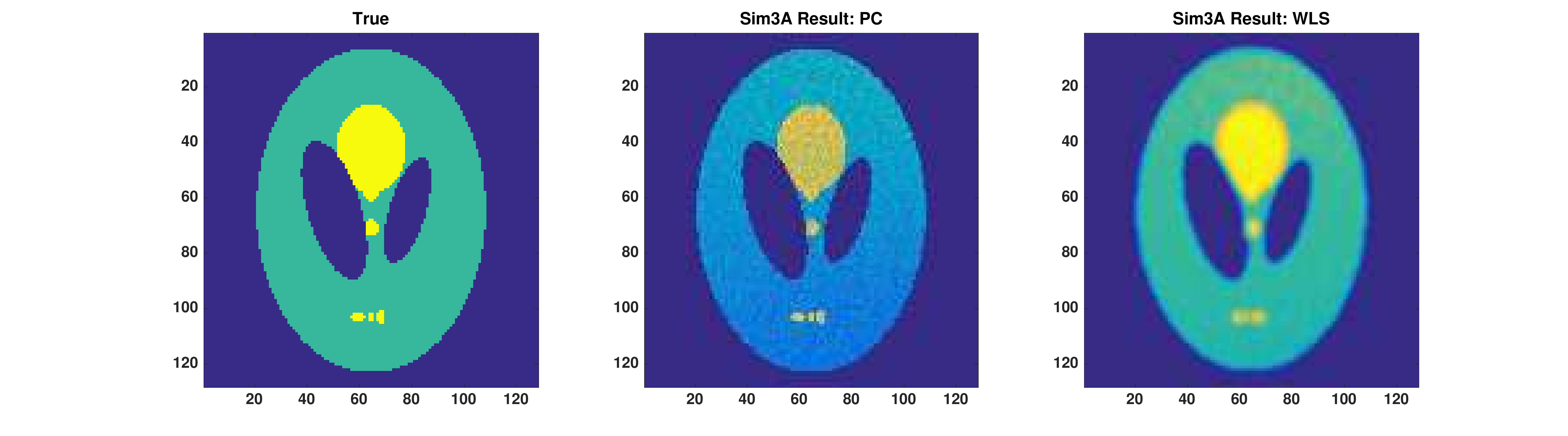}
       \includegraphics[width=6in]{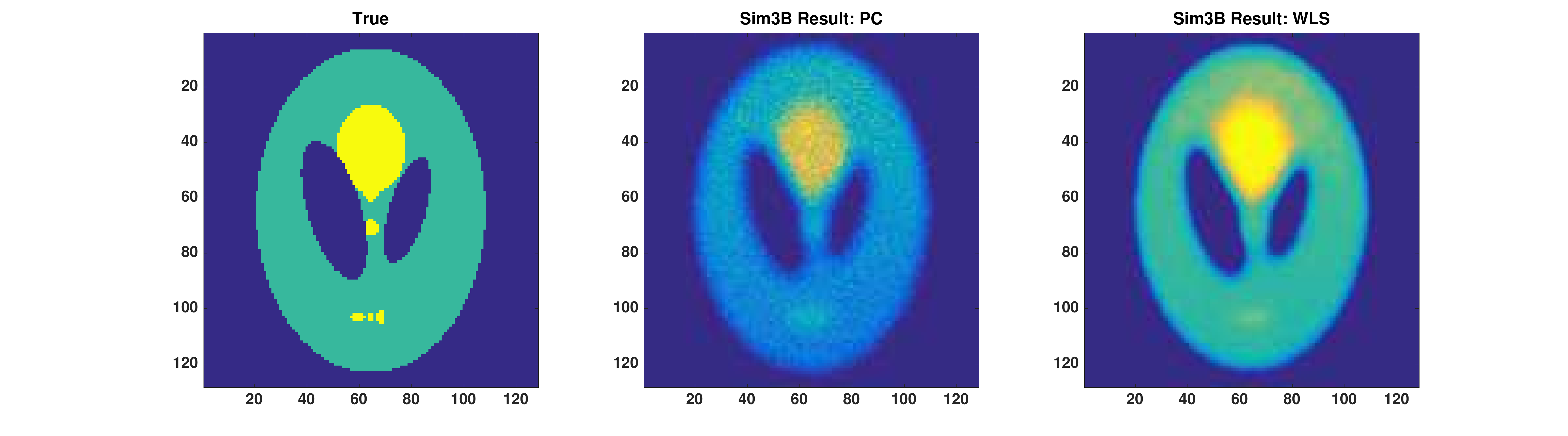}
       \includegraphics[width=6in]{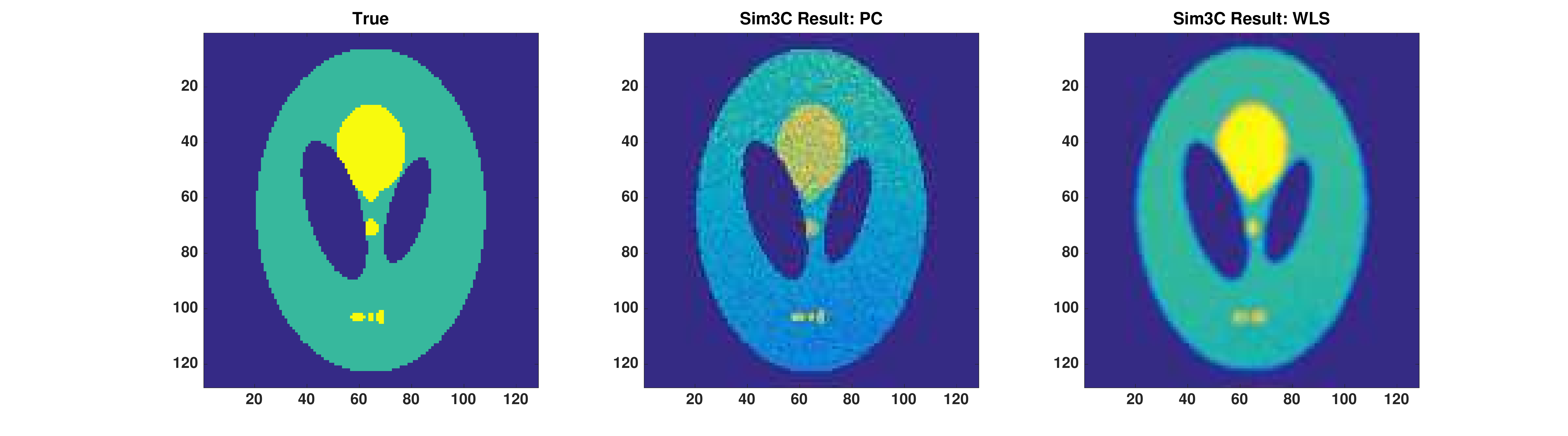}
       \includegraphics[width=6in]{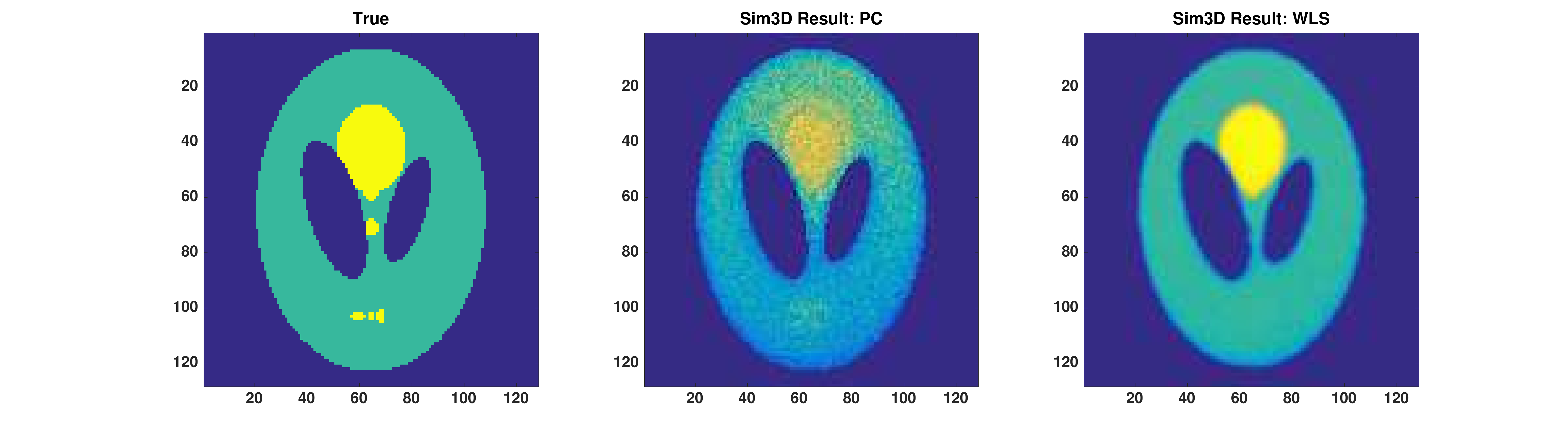}
       \caption{\label{Sim3} The true class labels and the averaged predicted labels in simulation 3. }
\end{figure}

Table \ref{Sim3out} indicates that clustering the WLS estimates performs significantly better in all three regions in terms of recall rates and precision rates under these four settings. Figure \ref{Sim3} shows the averaged predicted class labels of both strategies, and indicates that the results of clustering based on the WLS estimates are very consistent and with small variations. As the correct clustering rate of a single location has been modified for 3B and 3D, it might not be easy to compare the results between 3B and 3A, and those between 3D and 3C directly. The intuition, however, is that around the boundaries the signals become weaker after smoothing and it is easier to separate the pixels away from the boundaries to correct groups, and simultaneously around the boundaries between groups, the cost of mis-classification becomes lower and thus the precision rates are higher given the recall remains similar. The phenomenon is observed for both approaches.

\begin{table}
\caption{\label{Sim3out} The summarized statistics of Simulation 3, where Recall $= \frac{TP}{TP+FN}$, and  Precision $= \frac{TP}{TP+FP}$. TP, FP, FN stand for true positive, false positive and false negative, respectively.}
\centering\scriptsize
\begin{tabular}{cccccccccc} \hline
\multirow{2}{*}{Method} & \multirow{2}{*}{Cluster}  & \multicolumn{2}{c}{3A} & \multicolumn{2}{c}{3B} & \multicolumn{2}{c}{3C} &\multicolumn{2}{c}{3D} \\ \cline{3-10}
 &  &  Recall  & Precision &  Recall  & Precision & Recall  & Precision &  Recall  & Precision  \\ \hline
 & $\mathcal{S}_0 $  &  .983 (.002) & .970 (.005) & .973 (.003) & .987 (.002) &  .985 (.002) & .968 (.005) & .998 (.001) & .974 (.002) \\
WLS & $\mathcal{S}_1 $  &  .919 (.017) & .954 (.005) & .893 (.023) & .962 (.003) & .943 (.009) & .963 (.005) & .915 (.009) & .987 (.002) \\ 
 & $\mathcal{S}_2 $  &  .859 (.042) & .763 (.086) & .890 (.033) & .926 (.016) &  .884 (.027) & .939 (.015) & .889 (.026) & .987 (.007) \\ \hline
 & $\mathcal{S}_0 $  &  1.000 (.000) & .667 (.011) & .992 (.002) & .941 (.002) &  1.000 (.000) & .663 (.028) & 1.000 (.000) & .942 (.006)  \\
PC$^2$ & $\mathcal{S}_1 $  & .238 (.028) & .832 (.012) & .231 (.027) & .939 (.004) &  .232 (.062) & .881 (.018) & .247 (.063) & .977 (.005) \\ 
 & $\mathcal{S}_2 $  & .149 (.031) & .469 (.076) &  .150 (.032) & .914 (.020) &  .143 (.027) & .638 (.380) & .126 (.034) & .885 (.094) \\ \hline
\end{tabular}
\end{table}

\section{DATA ANALYSIS: HCP RESTING STATE fMRI DATA}\label{sec:application}

Functional Magnetic Resonance Imaging (fMRI) is one of the mainstays of brain research in that it allows \textit{in-vivo} detection of Blood Oxygenation Level Dependent (BOLD) signals \citep{Ogawa1990} describing cortical activation. The analysis of such signals has made it possible to determine which anatomical locations of the brain activate when a specific task is performed. More recently, it has become of central interest to use BOLD signals to determine how the different parts of the brain interact, i.e. infer the brain functional network organization. In fact, brain function is characterized by long-range interactions of distinct, local, and highly-specialized areas \citep{Tononi1994,Eickhoff2018}. Functional imaging data are routinely adopted to infer these interacting regions, which can help us understand brain organization and function by characterizing every brain location with its functional properties. Methodologically, it is essential that we take advantage of these functional properties to estimate long-range interactions while also accommodating the presence of localized units by incorporating spatial/anatomical information. Such studies are also referred to as functional connectivity studies.
In this section, we investigate functional connectivity by applying the eigen-adjusted FPCA framework introduced and we compare its results with a popular approach to connectivity analysis.

\begin{figure}[!htb]
\centering
\includegraphics[width=0.65\textwidth]{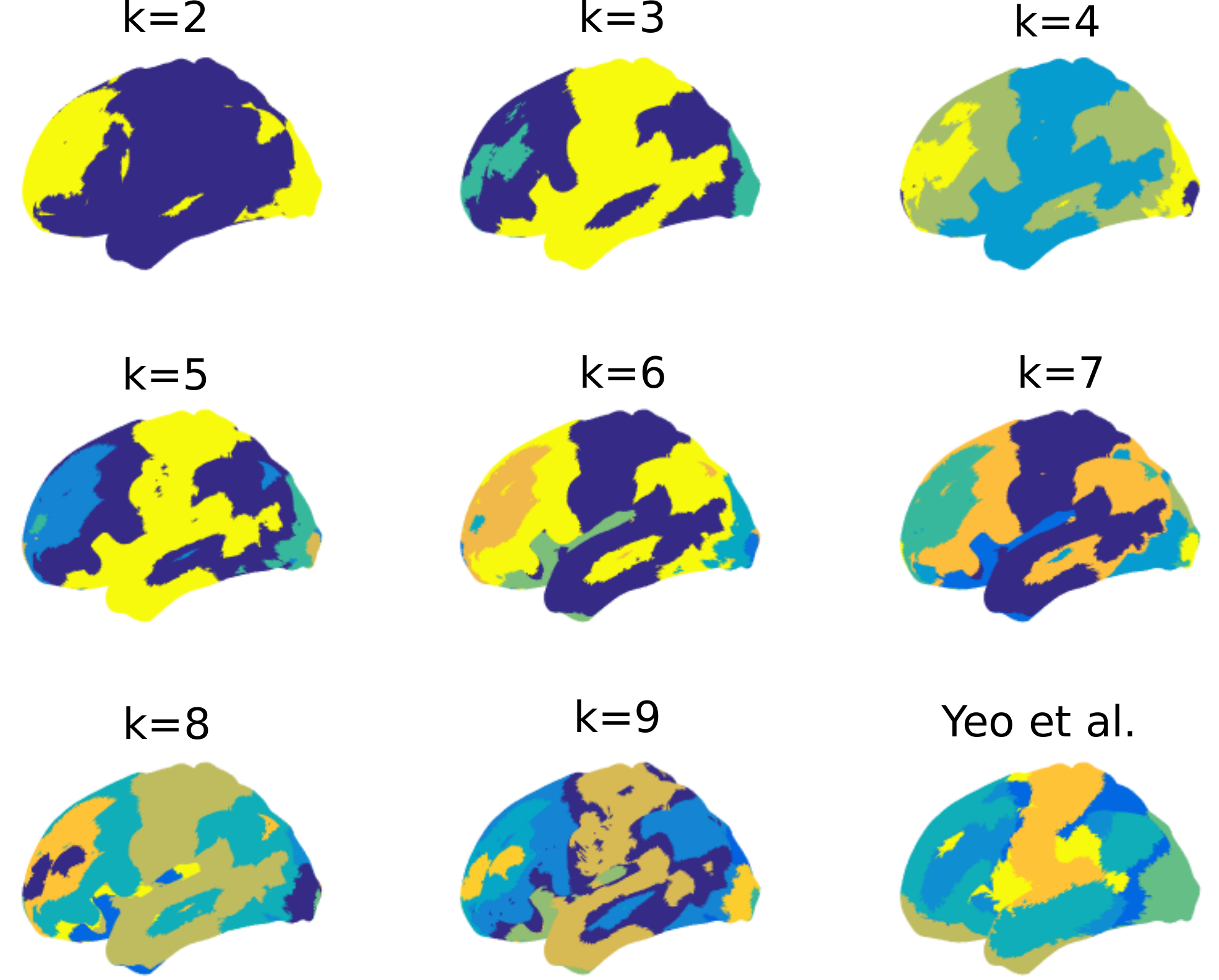}
\caption{Here we show the results of the $k$-means clustering applied to the WLS estimates of the eigenvalues, for different choices of $k$. On the bottom-right panel, the parcellation obtained in \cite{Yeo2011}. For instance, for $k=3, \ldots, 8$, we can see that the proposed model is able to separate the motor cortex (central part of the cerebral cortex) from the rest of the cerebral cortex. The visual cortex (bottom right) is also identified as a separate cluster. For higher numbers of clusters, e.g. $k=8$, a sub-cluster isolating the primary visual cortex is also identified.}
\label{fig:clustering_HCP}
\end{figure}

Data for this application consist of fMRI scans of 40 unrelated healthy subjects, collected within the Human Connectome Project \citep[HCP, ][]{Essen2012}. The minimal preprocessing pipeline has been applied to the dataset \citep{Glasser2013}. This includes artifact removal, motion correction, and registration to a standard space. The relevant fMRI signals arise from the cerebral cortex, which is the outermost layer of the brain. Here we adopt a distortion minimizing 2D planar parameterization of the cerebral cortex to apply the eigen-adjusted FPCA model, and visualize the results on the cortical surface. In our analysis, we consider the first 144 seconds of the fMRI signal detected under resting-state conditions, i.e., without requiring the subjects to perform any specific task. The time samples are acquired at regular intervals of 0.72 seconds. This results in a noisy signal $Y_s(t,\bz)$, for the $s$-th subject, observed at a regularly spaced grid of times $t_1, \ldots, t_N$, with $N = 200$, and locations $\bz_1, \ldots, \bz_{N_v} \subset \bbR^2$, with $N_v$ = 32,000.

\begin{figure}[!htb]
\centering
\includegraphics[width=0.6\textwidth,angle=270]{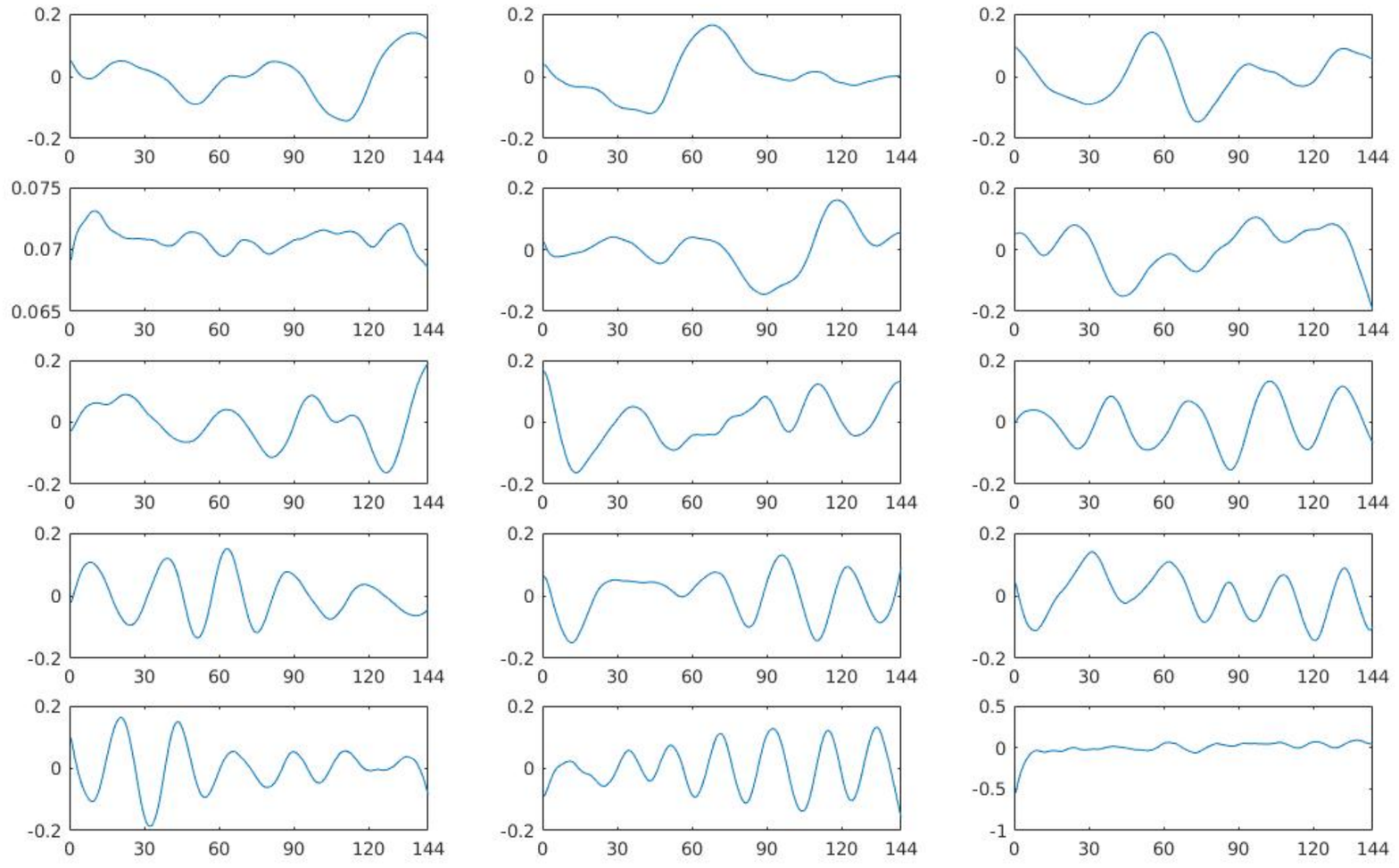}
\caption{From left-to-right then top-to-bottom, we show the resulting first $L = 15$ eigenfunctions $\hat\phi_k(t)$ estimated from the pooled empirical covariance function $\hat{\Gamma}^*(s,t)$.}
\label{fig:eigf_HCP}
\end{figure}

We perform eigen-adjusted FPCA on the resting-state fMRI data as follows. We apply a 3D (2D spatial $+$ 1D temporal) Nadaraya-Watson smoother, with product kernel, to estimate the mean function and use $5$-fold cross-validation to select their bandwidths. Cross-validation suggests very large bandwidths for the time smoothing, and specifically, a zero-mean function. This is not surprising as the location-specific mean, computed across time, has been removed in the preprocessing phase.

We then proceed with the estimation of the covariance functions $\Gamma(s, t, \bz)$, from the eigen-adjusted FPCA model in (\ref{eq:cov}), for each spatial location $\bz$. Specifically, the pooled covariance function $\Gamma^*(s,t)$ is estimated by applying a 2D local linear smoother to the associated empirical covariance. The bandwidth is chosen by $5$-fold cross-validation. We then compute the eigenfunctions $\hat\phi_k(t)$, shown in Figure~\ref{fig:eigf_HCP}, by applying eigen-analysis on a discretized version of the estimated covariance function $\hat{\Gamma}^*(s,t)$. We estimate the eigenvalue maps $\hat\lambda_j(\bz)$, $j=1,\ldots,L$ by means of WLS, as described in (\ref{eq:lda_wls}). We keep the first $L=15$ eigenfunctions, which explain more than $90\%$ of the total variance of the pooled covariance. We manually set $h_\lambda = 10$ to reflect the amount of smoothing imposed in the functional connectivity study in \cite{Yeo2011}, which we use to compare our results. For each location $\bz$, the resulting eigenvalues vector
\[
(\hat\lambda_1(\bz), \ldots, \hat\lambda_{L}(\bz))
\]
represents a multivariate summary of the temporal covariance structure $\Gamma(\cdot, \cdot, \bz)$. We finally identify co-activating brain regions by clustering the location-specific covariance functions through their associated eigenvalues estimates, and specifically, by applying k-means clustering on them, for different choices of $k$. The results are shown in Figure~\ref{fig:clustering_HCP}.

Resting-state time-series cannot be aligned to each other in any sensible way, thus eigenfunctions that describe localized (in time) modes of variation should instead be ascribed to experimental artifacts. The visualization of the estimated eigenfunctions offers a diagnostic tool to assess which are the modes of variation that for this reason should be removed from the subsequent analysis. From Figure~\ref{fig:eigf_HCP}, we can see that the $2$-nd and $15$-th eigenfunctions do indeed capture a localized mode of variation. However, removing them from the subsequent analysis did not substantially change the resulting parcellations.

Moreover, note that in this application, the value at each spatial location in $\bz_1, \ldots, \bz_{N_v}$ can be observed for every subject. Therefore, the spatial location $\bm{\bz}$ is not random. Specifically, the pooled covariance function becomes $\Gamma^*(s,t) = \sum_k \bar{\lambda}_k \phi_k(s)\phi_k(t)$, where $\bar{\lambda}_k = \sum_{v=1}^{N_v} \lambda_k(\bm{\bz}_v)/N_v.$
Therefore, if none of $\bar{\lambda}_k$'s are tied, the eigenfunctions can be consistently estimated.

In Figure~\ref{fig:clustering_HCP}, we also show the functional parcellation obtained in \cite{Yeo2011}, where instead the following approach is adopted. A set of $1,175$ uniformly sampled Regions of Interest (ROI) vertices on the cerebral cortex is identified. For each subject, Pearson correlation between the fMRI time series at each spatial location $\bz$ and that on the ROI vertices are computed, so that each spatial location is characterized by its first-order dependency to the ROIs. Only the top 10\% correlation values for each subject are kept and these are binarized and averaged across subjects. For each vertex $\bz$, this results in a multivariate descriptor
\[
(\rho(\bz,\text{ROI}_1), \ldots, \rho(\bz,\text{ROI}_{\text{1,175}})),
\]
where $\rho(\bz_1, \bz_2)$ is the correlation between the time-series at the two vertices $\bz_1, \bz_2$. Finally, spatial clustering is applied on the multivariate descriptors obtained by averaging those across subjects. In Figure~\ref{fig:clustering_HCP}, bottom-right panel, we show the resulting 7-networks parcellation.

Both the presented approaches aim at constructing multivariate summaries of the connectivity at each vertex, although in different fashions. To this purpose, the eigen-adjusted FPCA approach constructs a descriptor based on the temporal covariance structure at each location. Instead, the approach in \cite{Yeo2011}, constructs a correlation descriptor based both on time and locations. It could be argued that in the latter approach, given that spatial information is used to construct both the multivariate descriptors and to perform spatial clustering, you would expect the clustering to be influenced by the choice of the ROIs. This cannot happen with the eigen-adjusted FPCA approach, where the multivariate descriptor exploits exclusively the temporal component, with the exception of the spatial smoothing effect introduced to contrast the low signal-to-noise ratio.

Despite the two approaches being different, and thus not immediately comparable, their results are compatible. Consider, for instance, the eigen-adjusted FPCA results in Figure~\ref{fig:clustering_HCP}, for $k=7$. The proposed model is able to separate the motor cortex (central part of the cerebral cortex) from the prefrontal cortex (left part of the cerebral cortex). The visual cortex (bottom right) is also identified as a separate cluster, with a sub-cluster that seems to isolate the primary visual cortex. This separation of the brain into distinct clusters determines brain networks in a completely data-driven way, without the need for a-priori explicit or implicit spatial assumptions.

\section{Discussions}\label{sec:discussions}

We have demonstrated that in a number of situations the eigen-adjusted FPCA approach is able to provide a computationally efficient yet flexible alternative to either a simple mean-adjusted model or a fully-adjusted FPCA approach in the presence of covariates. Consistent estimates can be obtained using a WLS approach, without the need to estimate the full spatially varying covariance function but rather by a pooled version. It has been shown in simulations that the approach is effective in finite samples, particularly in relation to the kinds of data available in brain imaging applications. The application to functional connectivity shows that a comparison can be made based on the full covariate information across space, without the need to either reduce the dimension through an a-priori spatial downsampling (such as definitions of ROIs) or via a seed-based approach of choosing one or two locations to compare the full data set against.

There are a number of limitations to the approach. Firstly, should the eigenfunctions themselves directly vary with the covariates, then the model will only ever provide an approximation. However, as we have seen in the simulations, the eigen-adjusted covariance approximation is considerably better than a more simplistic approach of using a pooled covariance function in this case. In addition, it is possible that the WLS estimator could produce negative estimates, in finite samples, for positive eigenvalues. Should this be an issue, a number of possible recourses exist, including truncation or non-negative least-squares approaches. However, for positive eigenvalues, the asymptotic properties show that this will only ever be a finite sample problem.

Overall, the eigen-adjusted FPCA approach provides a set of tools to investigate functional covariance structures that include covariate information. It is likely that as more applied questions become framed in terms of second-order structure, as seen in functional connectivity, techniques such as these will only be further needed and utilized.

\section*{Acknowledgements}
The authors would like to express gratitude for the valuable comments of an anonymous referee, the Associate Editor and the Editor.

\appendix

\section{Assumptions}\label{app:assum}


\bigskip
The estimators $\hat\mu$ and $\hat\Gamma^*$ have been constructed by the local linear smoothing method. Therefore, it is natural to make the standard smoothness assumptions on the second derivatives of $\mu$ and $\Gamma^*$. Assume that the data
$(\mathbf{T}_i,\bZ_i,\mathbf{Y}_i)$, $i=1,\ldots,n$, have the same distribution, where $\mathbf{T}_i=(T_{i1},\ldots,T_{iN_i})$ and
$\mathbf{Y}_i=(Y_{i1},\ldots,Y_{iN_i})$.  Notice that we assume $(T_{ij},\bZ_{i})$ has the marginal density $g(t,\bz)$. Additional assumptions and conditions are listed below.

\begin{itemize}
\item[A.1] For some constants $m_T>0$ and $M_T<\infty$, $m_T\leq g(t,\bz) \leq M_T$ for all $t\in\bbT$ and $\bz\in\bbZ$. Further, $g(\cdot,\cdot)$ is differentiable with a bounded derivative.

\item[A.2] The kernel function $K(\cdot)$ is a symmetric probability density function on $[-1,1]$ and is of bounded variation on $[-1,1]$. Further, we denote $\nu_2 = \int_{-1}^1 u^2K(u)du$.



\end{itemize}

The following assumptions are about $Y(t)$ and were also made in \cite{LiH:10:1}.  Suppose the observation of the $i$th subject at time $T_{ij}$ is $Y_{ij} = \mu(T_{ij},\bZ_i)+U_{ij}$, where $\text{cov}(U_i(s),U_i(t))=\Gamma(s,t,\bz_i)+\sigma^2I(s=t)$ and $\Gamma(s,t,\bz_i)=\sum_{\ell} \lambda_\ell(\bz_i)\phi_\ell(s)\phi_\ell(t)$. Let $\bh_z = (h_z^{(1)},\ldots,h_z^{(p)})^T$ and denote $|\bh_z| = \prod_{i=1}^p h_z^{(i)}$. Also, $\gamma_{nk} = \left( n^{-1}\sum_{i=1}^n N_i^{-k} \right)^{-1}$ for $k=1$, and $2$. 

\begin{itemize}

\item[B.1] $\mu$ is twice differentiable and the second derivative is bounded on $\bbT\times \mathbb{\bZ}$.

\item[B.2] $\Ex(|U_{ij}|^{\lambda_\mu}) < \infty$ and $\Ex(\sup_{t\in\bbT} |X(t)|^{\lambda_\mu}) < \infty$ for some $\lambda_\mu \in (2,\infty)$; $h_\mu \rightarrow 0$ and $(h^2_\mu|\bh_z|+h_\mu|\bh_z|/\gamma_{n1})^{-1} (\log n/n)^{1-2/\lambda_\mu} \rightarrow 0$ as $n \rightarrow \infty$.

\item[B.3] All second-order partial derivatives of $\Gamma^*$ exist and are bounded on
$\bbT\times\bbT$.

\item[B.4] $\Ex(|U_{ij}|^{2\lambda_\Gamma}) < \infty$ and $\Ex(\sup_{t\in\bbT} |X(t)|^{2\lambda_\Gamma}) < \infty$ for some $\lambda_\Gamma \in (2,\infty)$; $h_\Gamma \rightarrow 0$ and $(h^4_\Gamma+h^3_\Gamma/\gamma_{n1} + h^2_\Gamma/\gamma_{n2})^{-1}(\log n/n)^{1-2/\lambda_\Gamma} \rightarrow 0$ as $n \rightarrow \infty$

\item[B.5] All second-order partial derivatives of $\lambda_k(\bz)$ exist and are bounded on
$\bbZ$ for $1\leq k \leq L$.

\item[B.6] $\bh_\lambda \rightarrow 0$ and $(|\bh_\lambda|)^{-1}(\log n/n)^{1-2/\eta} \rightarrow 0$ as $n \rightarrow \infty$ for some $\eta \in (2,\infty)$.

\item[B.7] For $\bz\in\bbZ$ and all $j$, $0<\lambda_j(\bz)/\lambda_j^*<\infty$.

\end{itemize}

The following assumptions are for the lemmas. Let $\bd_n = (d_n^{(1)},\ldots,d_n^{(p)})^T$.

\begin{itemize}
\item[C.1] $\Ex(|U|^\lambda)<\infty$ and $\Ex(\sup_{t\in\bbT,\bz\in\mathbb{\bZ}}|X(t,\bz)|^\lambda)<\infty \text{ for some } \lambda\in(2,\infty).$

\item[C.2] Let $c_n$ and $d_n^{(i)}$ for $i=1,\ldots,p$ be positive sequences tending to 0, $\beta_n = c_n^2|\bd_n| + c_n|\bd_n|/\gamma_{n1}$ and $\beta_n^{-1}(\log n/n)^{1-2/\lambda}=o(1).$

\item[C.2'] Let $c_n$ be a positive sequence tending to 0, $\beta_n = c_n^2 + c_n/\gamma_{n1}$ and $\beta_n^{-1}(\log n/n)^{1-2/\lambda}=o(1).$

\item[C.3] $\Ex(|U|^{2\lambda})<\infty$ and $\Ex(\sup_{t\in\bbT,\bz\in\mathbb{\bZ}}|Y(t,\bz)|^{2\lambda})<\infty \text{ for some } \lambda\in(2,\infty).$

\item[C.4] Let $c_n$ be a positive sequence tending to 0, $\beta_n = c_n^4 + c_n^3/\gamma_{n1}+ c_n^2/\gamma_{n2}$ and $\beta_n^{-1}(\log n/n)^{1-2/\lambda}=o(1).$

\end{itemize}

\section{A PC-based Approach for $\lambda_k(\bz)$ Estimation}\label{app:PC}

Besides the WLS approach, one intuition is to first predict the PC scores and apply a $p$-dimensional smoother to the squared PC scores given that $\Ex\{A_k^2(\bz)\}=\lambda_k(\bz)$. When the data are dense, the PC scores can be predicted well via a numerical approach for integration.
Here, we employ the trapezoidal rule for integration. Specifically,
\begin{equation}\label{Eq:PCE}
\hat A_{ik} = \sum_{j=1}^{N_i-1} \left[\hat U_{ij}\hat\phi_k(t_{i,j}) + \hat U_{i,j+1}\hat\phi_k(t_{i,j+1})\right] \frac{(t_{i,j+1}-t_{ij})}{2}.
\end{equation}
Then, a $p$-dimensional smoother can be applied to $\{(\hat A_{ik}^2,\bz_i)|i=1,\ldots,n\}$ to consistently estimate $\lambda_k(\bz)$ given that $\hat A_{ik}^2$ is a consistent estimator of $\lambda_k(\bz_i)$. 
Let
\begin{equation*}
\tilde{A}_{ik}  = \int_\mathcal{T} \{X_i(t)-\mu(t,\bz_i)\}\hat\phi_k(t)dt \text{, and }
\hat{A}^*_{ik}  = \int_\mathcal{T} \{Y_i(t)-\hat\mu(t,\bz_i)\} \hat\phi_k(t)dt.
\end{equation*}
The asymptotics of $\hat{A}_{ik}$ can be obtained by employing the inequality, $$|\hat{A}_{ik}-A_{ik}| \leq |\hat{A}_{ik}-\hat{A}^*_{ik}| + |\hat{A}^*_{ik}-\tilde{A}_{ik}|+|\tilde{A}_{ik}-A_{ik}|.$$ First, $|\tilde{A}_{ik} - A_{ik}|=O(h_\Gamma^2+\delta_{n1}+h_t^4+h_z^4+\delta_n^2) \ a.s.$ by applying Theorem \ref{A:thm3}. Second, $|\hat{A}_{ik} - \hat{A}^*_{ik} | = O(1/N_i^4) \ a.s.$ since $\hat{A}^*_{ik}$ is an approximation by the trapezoidal rule. Last, $|\hat{A}^*_{ik}-\tilde{A}_{ik}|=O((\log N_i/N_i)^{1/2} + \delta_n+h_t^2+h_z^2) \ a.s.$ by Lemma 5 in \cite{LiH:10:1} and Theorem \ref{A:thm1}.
Therefore, $|\hat{A}_{ik} - A_{ik}|=O((\log N_i/N_i)^{1/2} + h_\Gamma^2+\delta_{n1}+h_t^2+h_z^2+\delta_n)$ $a.s.$. Since $\hat{A}_{ik}$ is a consistent estimator, $\lambda_k(\bz)$ can be estimated by applying a $p$-dimensional smoother to $\{(\hat A_{ik}^2,\bz_i)|i=1,\ldots,n\}$. In the numerical studies, a local linear smoother is used. Specifically,
\begin{align}
 \hat\lambda_k(\bz) & = \hat b_0,  \text{ where } \label{eq:lda_in} \\ \nonumber & (\hat b_0,\hat{\bm{b}}_1^T)^T  = \arg\min_{\bm{b}} \frac{1}{n} \sum_{i=1}^{n} \left\{ \hat{A}^2_{ik}-b_0-\bm{b}_1^T(\bz_i-\bz) \right\}^2 \Big( \prod_{k=1}^p K_{h_\lambda^{(k)}}(z_i^{(k)}-z^{(k)})\Big).
\end{align}

\section*{SUPPLEMENTARY MATERIALS}
\begin{description}

\item[Supplement] Assumptions, Lemmas and Proofs and additional simulation (*.pdf)

\end{description}

\bibliography{fda}
\bibliographystyle{chicago}

\end{document}